\documentclass[runningheads]{llncs}
\usepackage[T1]{fontenc}
\usepackage{graphicx}
\usepackage{amsmath}
\usepackage{amssymb}
\usepackage{thmtools}
\usepackage{hyperref}
\usepackage{cleveref}
\usepackage{color}
\usepackage{subcaption}

\usepackage[ruled,linesnumbered,vlined]{algorithm2e}
\SetAlgorithmName{Algorithm}{Algorithm}{List of algorithms}

\newcommand{\OO}[1]{O( {#1} ) }

\newcommand{\ycert}{\texttt{YES}-certificate}
\newcommand{\ncert}{\texttt{NO}-certificate}
\newcommand{\yncert}{\texttt{YES}- and \ncert s}
\newcommand{\sort}{\textsc{sort}}
\newcommand{\scan}{\textsc{scan}}
\newcommand{\LCA}{\text{LCA}_T}

\begin{document}
\title{Certifying Induced Subgraphs in Large Graphs}
%
%
\author{Ulrich Meyer\inst{1} \and Hung Tran\inst{1} \and Konstantinos Tsakalidis\inst{2}}
%
\authorrunning{U. Meyer et al.}
%
\institute{Goethe University Frankfurt, Germany \\ \email{\{umeyer,htran\}@ae.cs.uni-frankfurt.de} \and	University of Liverpool, United Kingdom \\ \email{K.Tsakalidis@liverpool.ac.uk} }
%
\maketitle              
\begin{abstract}
We introduce I/O-optimal certifying algorithms for bipartite graphs, as well as for the classes of split, threshold, bipartite chain, and trivially perfect graphs. When the input graph is a class member, the certifying algorithm returns a certificate that characterizes this class. Otherwise, it returns a forbidden induced subgraph as a certificate for non-membership. On a graph with $n$ vertices and $m$ edges, our algorithms take optimal $\OO{\sort(n + m)}$ I/Os in the worst case or with high probability for bipartite chain graphs, and the certificates are returned in optimal I/Os. We give implementations for split and threshold graphs and provide an experimental evaluation.
\end{abstract}

\section{Introduction}

\emph{Certifying algorithms} \cite{MMNS11} ensure the correctness of an algorithm's output without having to trust the algorithm itself. The user of a certifying algorithm inputs~$x$ and receives the output $y$ with a \emph{certificate} or \emph{witness} $w$ that proves that~$y$ is a correct output for input $x$. Certifying the bipartiteness of a graph is a textbook example where the returned witness $w$ is a bipartition of the vertices (\ycert) or an induced \emph{odd-length cycle} subgraph, i.e. a cycle of vertices with an odd number of edges (\ncert). 

Emerging big data applications need to process large graphs efficiently. Standard models of computation in internal memory (RAM, pointer machine) do not capture the 
algorithmic complexity of processing graphs with size that exceed the main memory. 
The \emph{I/O model} by Aggarwal and Vitter~\cite{AV88} is suitable for studying large graphs stored in an external memory hierarchies, e.g. comprised of cache, RAM and hard disk memories. 
The input data elements are stored in \emph{external memory} (EM) packed in \emph{blocks} of at most $B$ elements and computation is free in \emph{main memory} for at most $M$ elements. The \emph{I/O-complexity} is measured in \emph{I/O-operations} (\emph{I/Os}) that transfer a block from external to main memory and vice versa. 
\emph{I/O-optimal} external memory algorithms for sorting and scanning $n$ elements take \sort$\left(n\right) = \OO{(n/B)\log_{M/B}(n/B)}$ I/Os and \scan$\left(n\right) = \OO{n/B}$ I/Os, respectively.

\subsection{Previous Work}

Certifying bipartiteness in internal memory takes time linear in the number of edges by any traversal of the graph. 
However, all known external memory breadth-first search~\cite{AM09} and depth-first search~\cite{BGVW00} traversal algorithms take suboptimal $\omega\left(\text{\sort}\left(n+m\right)\right)$ I/Os for an input graph with $n$ vertices and $m$ edges. 
Heggernes and Kratsch~\cite{HK07} present optimal internal memory algorithms for certifying whether a graph belongs to the classes of split, threshold, bipartite chain, and trivially perfect graphs. 
They return in linear time a \ycert{}~characterizing the corresponding class or a forbidden induced subgraph of the class (\ncert{}).
The \yncert{} are authenticated in linear and constant time, respectively.
A straightforward application to the I/O model leads to suboptimal certifying algorithms since graph traversal algorithms in external memory are much more involved and no worst-case efficient algorithms are known.


\subsection{Our Results} 
We present I/O-optimal certifying algorithms for \emph{bipartite}, \emph{split}, \emph{threshold}, \emph{bipartite chain}, and \emph{trivially perfect} graphs.
All algorithms return in the membership case, a \ycert\ $w$ characterizing the graph class, or a $\OO{1}$-size \ncert\ in the non-membership case. As a byproduct, we show how to efficiently certify graph \emph{bipartiteness} in external memory using standard I/O-efficient techniques.
Additionally, we perform experiments for split and threshold graphs showing scaling beyond the size of main memory.

\section{Preliminaries and Notation}
\label{sec:preliminaries}

For a graph $G = (V, E)$, let $n = |V|$ and $m = |E|$ denote the number of vertices $V$ and edges $E$, respectively.
For a vertex $v \in V$ we denote by $N(v)$ the \emph{neighborhood} of $v$ and by $N[v] = N(v) \cup \{v\}$ the \emph{closed neighborhood} of $v$.
The \emph{degree} $\deg(v)$ of a vertex $v$ is given by $\deg(v) = |N(v)|$.
A vertex is called \emph{simplicial} if $N(v)$ is a clique and \emph{universal} if $N[v] = V$.

\subsubsection{Graph Substructures and Orderings}
The subgraph of $G$ that is induced by a subset $A \subseteq V$ of vertices is denoted by $G[A]$.
The \emph{substructure} (subgraph) of a cycle on $k$ vertices is denoted by $C_k$ and of a path on $k$ vertices is denoted by $P_k$.
The substructure $2K_2$ is a graph that is isomorphic to the following constant size graph: $(\{a,b,c,d\}, \{ab, cd\})$.

Henceforth we refer to different types of orderings of vertices: an ordering $(v_1, \ldots, v_n)$ is a (i) \emph{perfect elimination ordering} (\emph{peo}) if $v_i$ is simplicial in $G[\{ v_i, v_{i+1}, \ldots, v_n \}]$ for $i \in \{1, \ldots, n\}$, and a (ii) \emph{universal-in-a-component-ordering} (\emph{uco}) if $v_i$ is universal in its connected component in $G[ \{v_i, v_{i+1}, \ldots, v_n\} ]$ for $i \in \{1, \ldots, n\}$.
For a subset $X = \{v_1, \ldots, v_k\}$, we call $(v_1, \ldots, v_k)$ a \emph{nested neighborhood ordering} (\emph{nno}) if $(N(v_1) \setminus X) \subseteq (N(v_2) \setminus X)) \subseteq \ldots \subseteq (N(v_k) \setminus X)$.
Finally for any ordering, we partition $N(v_i)$ into lower and higher ranked neighbors, respectively, $L(v_i) = \{ x \in N(v_i) : v_i \text{ is ranked lower than } x \}$ and $H(v_i) = \{ x \in N(v_i) : v_i \text{ is ranked higher than } x \}$.

\subsubsection{Graph Representation}
We assume an \emph{adjacency row representation} where the graph $G =(V,E)$ is represented by two arrays $P = [\ P_i\ ]_{i=1}^n$ and $E = [\ u_i\ ]_{i=1}^m$.
The neighbors of a vertex $v_i$ are then given by the vertices at position $P[v_i]$ to $P[v_{i + 1}] - 1$ in $E$.
We use the adjacency row representation to allow for efficient scanning of $G$: (i) computing $k$ consecutive adjacency lists consisting of $m$ edges requires $\OO{\scan(m)}$ I/Os and (ii) computing the degrees of $k$ consecutive vertices requires $\OO{\scan(k)}$ I/Os.

\subsubsection{Time-Forward Processing} 
\emph{Time-forward processing} (\emph{TFP}) is a generic technique to manage data dependencies of external memory algorithms~\cite{MZ02}.
These dependencies are typically modeled by a directed acyclic graph $G = (V, E)$ where every vertex $v_i \in V$ models the computation of $z_i$ and an edge $(v_i, v_j) \in E$ indicates that $z_i$ is required for the computation of $z_j$.

Computing a solution then requires the algorithm to traverse $G$ according to some topological order $\prec_T$ of the vertices $V$.
The TFP technique achieves this in the following way: after $z_i$ has been calculated, the algorithm inserts a message $\langle v_j, z_i \rangle$ into a minimum priority-queue data structure for every successor $(v_i, v_j) \in E$ where the items are sorted by the recipients according to $\prec_T$.
By construction, $v_j$ receives all required values $z_i$ of its predecessors $v_i \prec_T v_j$ as messages in the data structure.
Since these predecessors already removed their messages from the data structure, items addressed to $v_j$ are currently the smallest elements in the data structures and thus can be dequeued with a delete-minimum operation. By using suitable external memory priority-queues~\cite{A03}, TFP incurs $\OO{\sort(k)}$ I/Os, where $k$ is the number of messages sent.

\section{Certifying Graphs Classes in External Memory}
\label{sec:small-subgraph-detection}

\subsection{Certifying Split Graphs in External Memory}
A split graph is a graph that can be partitioned into two sets of vertices $(K, I)$ where $K$ and $I$ induce a clique and an independent set, respectively.
The partition $(K, I)$ is called the \emph{split partition}.
They are additionally characterized by the forbidden induced substructures $2K_2, C_4$ and $C_5$, meaning that any vertex subset of a split graph cannot induce these structures \cite{HF77}.
Since split graphs are a subclass of chordal graphs, there exists a peo of the vertices for every split graph. 
In fact, any non-decreasing degree ordering of a split graph is a peo~\cite{HK07}.

Our algorithm adapts the internal memory certifying algorithm of Heggernes and Kratsch~\cite{HK07} to external memory by adopting TFP.
As output it either returns the split partition $(K, I)$ as a \ycert{} or one of the forbidden substructures $C_4, C_5$ or $2K_2$ as a \ncert.

First, we compute a non-decreasing degree ordering $\alpha = (v_1, \ldots, v_n)$ and relabel\footnote{If a vertex $v_i$ has rank $k$ in $\alpha$ it will be relabeled to $v_k$.} the graph according to $\alpha$.
Thereafter it checks whether $\alpha$ is a peo in $\OO{\sort(n + m)}$ I/Os by \autoref{prop:peo}.
In the non-membership case, the algorithm returns three vertices $v_j, v_k, v_i$ where $\{v_i, v_j\}, \{v_i, v_k\} \in E$ but $\{v_j, v_k\} \notin E$ and $i < j < k$, violating that $v_i$ is simplicial in $G[\{v_i, \ldots, v_n\}]$.
In order to return any of the forbidden substructures we find additional vertices that complete the induced subgraphs.
Note that $(v_k, v_i, v_j)$ already forms a $P_3$ and may extend to a $C_4$ if $N(v_k) \cap N(v_j)$ contains a vertex $z \ne v_i$ that is not adjacent to $v_i$.
Computing $(N(v_k) \cap N(v_j)) \setminus N(v_i)$ requires scanning the adjacencies of $\OO{1}$ many vertices totaling to $\OO{\scan(n)}$ I/Os.
If $(N(v_k) \cap N(v_j)) \setminus N(v_i)$ is empty we try to extend the $P_3$ to a $C_5$ or output a $2K_2$ otherwise.
To do so, we find vertices $x \ne v_i$ and $y \ne v_i$
for which $\{x, v_j\}, \{y, v_k\}\in E$ but $\{x, v_k\}, \{y, v_j\} \notin E$ that are also not adjacent to $v_i$, i.e.~$\{x, v_i\}, \{y, v_i\} \notin E$.
Both~$x$ and $y$ exist due to the ordering $\alpha$ \cite{HK07} and are found using $\OO{1}$ scanning steps requiring $\OO{\scan(n}$ I/Os.
If $\{x, y\} \in E$ then $(v_j, v_i, v_k, y, x)$ is a $C_5$, otherwise $G[\{v_j, x, v_k, y\}]$ constitutes a $2K_2$.
Determining whether $\{x, y\} \in E$ requires scanning $N(x)$ and $N(y)$ using $\OO{\scan(n)}$ I/Os.

In the membership case, $\alpha$ is a peo and the algorithm proceeds to verify first the clique $K$ and then the independent set $I$ of the split partition $(K, I)$.
Note that for a split graph the maximum clique of size $k$ must consist of the $k$-highest ranked vertices in $\alpha$ \cite{HK07} where $k$ can be computed using $\OO{\sort(m)}$ I/Os by \autoref{prop:maximumclique}.
Therefore, it suffices to verify for each of the $k$ candidates $v_i$ whether it is connected to $\{v_{i+1}, \ldots, v_n\}$ since the graph is undirected.
For a sorted sequence of edges relabeled by $\alpha$, we check this property using $\OO{\scan(m)}$ I/Os.
If we find a vertex $v_i \in \{ v_{n - k + 1}, \ldots, v_n\}$ where $\{v_i, v_j\} \notin E$ with $i < j$ then $G[\{v_i, \ldots, v_n\}]$ already does not constitute a clique and we have to return a \ncert.
Since the maximum clique has size $k$, there are $k$ vertices with degree at least $k - 1$.
By these degree constraints there must exist an edge $\{v_i, x\} \in E$ where $x \in \{v_1, \ldots, v_{i-1}\}$ \cite{HK07}.
Additionally, it holds that $\{x, v_j\} \notin E$ and there exists an edge $\{z, v_j\} \in E$ where $z \in \{v_1, \ldots, v_{i-x1}\}$ that cannot be connected to $x$, i.e.~$\{x, z\} \notin E$ \cite{HK07}.
Thus, we first scan the adjacency lists of $v_i$ and $v_j$ to find $x$ and $z$ in $\OO{\scan(n)}$ I/Os and return $G[\{v_i, v_j, x, z\}]$ as the $2K_2$ \ncert.
Otherwise let $K = \{ v_{n-k+1}, \ldots, v_n\}$.

Lastly, the algorithm verifies whether the remaining vertices form an independent set. 
We verify that each candidate $v_i$ is not connected to $\{v_{i+1}, \ldots, v_{n-k}\}$, since the graph is undirected.
For this, it suffices to scan over $n - k$ consecutive adjacency lists in $\OO{\scan(m)}$ I/Os.
More precisely, we scan the adjacency lists from $v_{n-k}$ to $v_1$ and in case an edge $\{ v_i, v_j\}$ where $i < j \le n - k$ is found we find two more vertices to again complete a $2K_2$.
For the first occurrence of such a vertex $v_i$, we remark that $\{v_{i + 1}, \ldots, v_{n-k}\}$ and $\{v_{n-k+1}, \ldots, v_n\}$ form an independent set and a clique, respectively.
Therefore there exists a vertex $y \in K$ that is adjacent to $x$ but not to $v_i$~\cite{HK07}.
We find $y$ by scanning $N(x)$ and $N(v_i)$ in $\OO{\scan(n)}$ I/Os.
To complete the $2K_2$ we similarly find $z \in N(y) \setminus (N(x) \cup N(y_i))$ in $\OO{\scan(n)}$ I/Os which is guaranteed to exist~\cite{HK07}.



\begin{algorithm}[t]
	\DontPrintSemicolon
	\caption{Recognizing Perfect Elimination in EM}
	\label{algo:peo}
	\KwData{edges $E$ of graph $G$, non-decreasing degree ordering $\alpha = (v_1, \ldots, v_n)$}
	\KwOut{bool whether $\alpha$ is a peo, three invalidating vertices $\{v_i, v_j, v_k\}$ if not}
	Sort $E$ and relabel according to $\alpha$ \\
	\For{$i = 1, \ldots, n$}{
		Retrieve $H(v_i)$ from $E$\\
		\If{$H(v_i) \ne \emptyset$}{
			Let $u$ be the smallest successor of $v_i$ in $H(v_i)$ \\
			\For{$x \in H(v_i) \setminus \{u\}$}{
				\textsf{PQ.push}($\langle u, x, v_i \rangle$) \tcp*{inform $u$ of $x$ coming from $v_i$}
			}
		}
		\While(\hfill \CommentSty{// for each message to $v_i$}){$\langle v, v_k, v_j \rangle \leftarrow$ {\normalfont \textsf{PQ.top}()} where $v = v_i$}{
			\If(\hfill \CommentSty{// $v_i$ does not fulfill peo property}){$v_k \notin H(v_i)$}{
				\Return \textsc{false}, $\{v_i, v_j, v_k\}$
			}
			\textsf{PQ.pop}()
		}
	}
	\Return \textsc{true}
\end{algorithm}

\begin{proposition}
	\label{prop:peo}
	Verifying that a non-decreasing degree ordering $\alpha = (v_1, \ldots, v_n)$ of a graph $G$ with $n$ vertices and $m$ edges is a perfect elimination ordering requires $\OO{\sort(n+m)}$ I/Os.
\end{proposition}
\begin{proof}
	We follow the approach of~\cite[Theorem 4.5]{G04} and adapt it to the external memory using TFP, see \autoref{algo:peo}.
	
	After relabeling and sorting the edges by $\alpha$ we iterate over the vertices in the order given by $\alpha$.
	For a vertex $v_i$ the set of neighbors $N(v_i)$ needs to be a clique.
	In order to verify this for all vertices, for a vertex $v_i$ we first retrieve $H(v_i)$.
	Then let $u \in H(v_i)$ be the smallest ranked neighbor according to $\alpha$.
	In order for $v_i$ to be simplicial, $u$ needs to be adjacent to all vertices of $H(v_i) \setminus \{u\}$.
	In TFP-fashion we insert a message $\langle u, w \rangle$ into a priority-queue where $w \in H(v_i) \setminus \{u\}$ to inform $u$ of every vertex it should be adjacent to.
	Conversely, after sending all adjacency information, we retrieve for $v_i$ all messages $\langle v_i, - \rangle$ directed to $v_i$ and check that all received vertices are indeed neighbors of $v_i$.
	
	Relabeling and sorting the edges takes $\OO{\sort(m)}$ I/Os.
	Every vertex $v_i$ inserts at most all its neighbors into the priority-queue totaling up to $\OO{m}$ messages which requires $\OO{\sort(m)}$ I/Os.
	Checking that all received vertices are neighbors only requires a scan over all edges since vertices are handled in non-descending order by $\alpha$. \hfill \qed
\end{proof}

\begin{algorithm}[t]
	\DontPrintSemicolon
	\caption{Maximum Clique Size for Chordal Graphs in EM}
	\label{algo:maximumclique}
	\KwData{edges $E$ of input graph $G$, peo $\alpha = (v_1, \ldots, v_n)$}
	\KwOut{maximum clique size $\chi$}
	Sort $E$ and relabel according to $\alpha$ \\
	$\chi \leftarrow 0$ \\
	\For{$i = 1, \ldots, n$}{
		Retrieve $H(v_i)$ from $E$ \tcp*{scan E}
		\If{$H(v_i) \ne \emptyset$}{
			Let $u$ be the smallest successor of $v_i$ in $H(v_i)$ \\
			\textsf{PQ.push}($\langle u, |H(v_i)| - 1 \rangle$) \tcp*{$v_i$ simplicial $\Rightarrow$ $G[N(v_i)]$ is clique}
		}
		$S(v_i) \leftarrow -\infty$ \\
		\While{$\langle v, S \rangle \leftarrow$ {\normalfont \textsf{PQ.top}()} where $v = v_i$}{
			$S(v_i) \leftarrow \max \{ S(v_i), S \}$ \tcp*{compute maximum over all}
			\textsf{PQ.pop}()
		}
		$\chi \leftarrow \max\{ \chi, S(v_i) \}$
	}
	\Return{$\chi$}
\end{algorithm}

\begin{proposition}
	\label{prop:maximumclique}
	Computing the size of a maximum clique in a split graph requires $\OO{\sort(m)}$ I/Os.
\end{proposition}
\begin{proof}
	Note that split graphs are both chordal and co-chordal \cite{HF77}.
	For chordal graphs, computing the size of a maximum clique in internal memory takes linear time \cite[Theorem 4.17]{G04} and is easily convertible to an external memory algorithm using $\OO{\sort(m)}$ I/Os.
	To do so, we simulate the data accesses of the internal memory variant using priority-queues to employ TFP, see \autoref{algo:maximumclique}.
	Instead of updating each $S(v_i)$ value immediately, we delay its consecutive computation by sending a message $\langle v_i, S \rangle$ to $v_i$ to inform $v_i$, that $v_i$ is part of a clique of size $S$.
	After collecting all messages, the overall maximum is computed and the global value of the currently maximum clique is updated if necessary. \hfill \qed
\end{proof}

By the above description it follows that split graphs can be certified using $\OO{\sort(n + m)}$ I/Os which we summarize in \autoref{lem:split}.

\begin{lemma}
	\label{lem:split}
	A graph with $n$ vertices and $m$ edges stored in external memory is certified whether it is a split graph or not in $\OO{\sort(n + m)}$ I/Os.
	In the membership case the algorithm returns in $\OO{\scan(K + I)}$ I/Os the split partition $(K, I)$ as the \ycert, and otherwise it returns a $\OO{1}$-size \ncert.
\end{lemma}

%
%

\subsection{Certifying Threshold Graphs in External Memory}

Threshold graphs \cite{C73,G04,MP95} are split graphs with the additional property that the independent set $I$ of the split partition $(K, I)$ has an nno.
Its corresponding forbidden substructures are $2K_2, P_4$ and $C_4$.
Alternatively, threshold graphs can be characterized by a graph generation process: repeatedly add universal or isolated vertices to an initially empty graph.
Conversely, by repeatedly removing universal and isolated vertices from a threshold graph the resulting graph must be the empty graph.
In comparison to certifying split graphs, threshold graphs thus require additional steps.

First, the algorithm certifies whether the input is a split graph.
In the non-membership case, if the returned \ncert\ is a $C_5$ we extract a $P_4$ otherwise we return the substructure immediately.
For the membership case, we recognize whether the input is a threshold graph by repeatedly removing universal and isolated vertices using the previously computed peo $\alpha$ in $\OO{\sort(m)}$ I/Os by \autoref{prop:universal-isolated}.
If the remaining graph is empty, we return the independent set~$I$ with its non-decreasing degree ordering.
Note that after removing a universal vertex $v_i$, vertices with degree one become isolated.
Since low-degree vertices are at the front of $\alpha$, an I/O-efficient algorithm cannot determine them on-the-fly after removing a high-degree vertex. Therefore pre-processing is required.
For every vertex $v_i$ we compute the number of vertices $S(v_i)$ that become isolated after the removal of $\{v_i, \ldots, v_n\}$.
To do so, we iterate over $\alpha$ in non-descending order and check for $v_i$ with $L(v_i) = \emptyset$.
Since $v_i$ has no lower ranked neighbors, it would become isolated after removing all vertices in $H(v_i)$, in particular when the successor with smallest index $v_j \in H(v_i)$ is removed.
We save $v_j$ in a vector \textsf{S} and sort \textsf{S} in non-ascending order.
The values $S(v_n), \ldots, S(v_1)$ are now accessible by a scan over \textsf{S} to count the occurrences of each $v_j$ in $\OO{\scan(m)}$ I/Os.

In the non-membership case, there must exist a $P_4$ since the input is split and cannot contain a $C_4$ or a $2K_2$.
We can delete further vertices from the remaining graph that cannot be part of a $P_4$.
For this, let $K' \subset K$ and $I' \subset I$ be the remaining vertices of the split partition.
Any $v \in K'$ where $N(v) \cap I' = \emptyset$ and any $v \in I'$ where $N(v) \cap K' = K'$ cannot be part of a $P_4$ \cite{HK07} and can therefore be deleted.
We proceed by considering and removing vertices of $K$ by non-descending degree and vertices of $I$ by non-ascending degree.
After this process, we retrieve the highest-degree vertex $v$ in $I$ where there exists $\{v, y\} \notin E$ and $\{y, z\} \in E$ where $y\in K$ and $z \in I$ \cite{HK07}.
Additionally, there is a neighbor $w \in K$ of $v$ for which $\{w, z\} \notin E$ \cite{HK07} and we return the $P_4$ given by $G[\{v,w,y,z\}]$. 
Finding the $P_4$ therefore only requires $\OO{\scan(n + m)}$ I/Os.


\begin{proposition}
	\label{prop:universal-isolated}
	Verifying that a non-decreasing degree ordering $\alpha = (v_1, \ldots, v_n)$ of a graph $G$ with $n$ vertices and $m$ edges emits an empty graph after repeatedly removing universal and isolated vertices requires $\OO{\sort(n) + \scan(m)}$ I/Os.
\end{proposition}
\begin{proof}
	Generating the values $S(v_n), \ldots, S(v_1)$ requires a scan over all adjacency lists in non-descending order and sorting \textsf{S} which takes $\OO{\sort(n) + \scan(m)}$ I/Os.
	Afte pre-processing, the algorithm only requires a reverse scan over the degrees $d_n, \ldots, d_1$, see \autoref{algo:p4threshold}.
	We iterate over $\alpha$ in reverse order, where for each $v_i$ we check whether $L(v_i) = \emptyset$.
	If $v_i$ is not isolated it must be universal.
	Therefore we compare its current degree $\deg(v_i)$ with the value $(n - 1) - n_{\text{del}}$ where $n_{\text{del}} = \sum_{j=j+1}^{n} S(v_j)$.
	All operations take $\OO{\scan(m)}$ I/Os in total. \hfill \qed
\end{proof}

We summarize our findings for threshold graphs in \autoref{lem:threshold}.

\begin{lemma}
	\label{lem:threshold}
	A graph with $n$ vertices and $m$ edges stored in external memory is certified whether it is a threshold graph or not in $\OO{\sort(n + m)}$ I/Os.
	In the membership case the algorithm returns in $\OO{\scan\left(\beta \right)}$ I/Os a nested neighborhood ordering $\beta$ as the \ycert, and otherwise it returns a $\OO{1}$-size \ncert.
\end{lemma}
\begin{proof}
	Certifying that the input graph is a split graph requires $\OO{\sort(n + m)}$ I/Os by \autoref{lem:split}.
	If it is, we check if the input is a threshold graph directly by checking whether the graph is empty after repeatedly removing universal and isolated vertices in $\OO{\sort(m)}$ I/Os by \autoref{prop:universal-isolated}. 
	Otherwise we have to find a $P_4$, since the input is a split but not a threshold graph. 
	Hence, this step requires  $\OO{\scan(n + m)}$ I/Os and the total I/Os are $\OO{\sort(n + m)}$. \hfill \qed
\end{proof}

\subsection{Certifying Trivially Perfect Graphs in External Memory}
\label{subsec:triviallyperfect}

Trivially perfect graphs have no vertex subset that induces a $P_4$ or a $C_4$ \cite{G04}.
In contrast to split graphs, any non-increasing degree ordering of a trivially perfect graph is a uco \cite{HK07}.
In fact, this is a one-to-one correspondence: a non-increasing sorted degree sequence of a graph is a uco iff the graph is trivially perfect \cite{HK07}.

\begin{algorithm}[t]
	\DontPrintSemicolon
	\caption{Recognizing Universal-in-a-Component Orderings in EM}
	\label{algo:uco}
	\KwData{edges $E$ of graph $G$, non-increasing degree ordering $\gamma = (v_1, \ldots, v_n)$}
	\KwOut{bool whether $\gamma$ is a uco}
	Sort $E$ and relabel according to $\gamma$ \\
	\For{$i = 1, \ldots, n$}{
		Vector \textsf{L} $= [0]$ \tcp*{initialize with 0}
		\While(\hfill \CommentSty{// $v_i$'s received labels}){$\langle v, v_j, \ell \rangle \leftarrow$ {\normalfont \textsf{PQ.top}()} where $v = v_i$}{
			\textsf{L.push}($\ell$) \\
			\textsf{PQ.pop}()
		}
		\For(\hfill \CommentSty{// \textsf{L.size} is even}){$i = 1, \ldots, ${\normalfont \textsf{L.size}}$/ 2$}{
			\If(\hfill \CommentSty{// mismatch / anomaly}){{\normalfont \textsf{L}[$2i$] $\ne$ \textsf{L}[$2i{+}1$] \textbf{and} \textsf{L.size} $> 1$}}{
				\Return \textsc{false}
			}
		}
		$\ell(v_i) \leftarrow \textsf{L}$[\textsf{L.size}] \tcp*{assign label of $v_i$}
		Retrieve $H(v_i)$ from $E$ \tcp*{scan E}
		\For{$u \in H(v_i)$}{
			\textsf{PQ.push}($ \langle u, v_i, \ell(v_i) \rangle $); \textsf{PQ.push}($ \langle u, v_i, i \rangle $) \\
		}
	}
	\Return \textsc{true}
\end{algorithm}

In external memory this can be verified using TFP by adapting the algorithm in \cite{HK07}, see \autoref{algo:uco}.
After computing a non-increasing degree ordering $\gamma$ the algorithm relabels the edges of the graph according to $\gamma$ and sorts them.
Now we iterate over the vertices in non-descending order of $\gamma$, process for each vertex $v_i$ its received messages and relay further messages forward in time.
 
Initially all vertices are labeled with $0$.
Then, at step $i$ vertex $v_i$ checks that all adjacent vertices $N(v_i)$ have the same label as $v_i$.
After this, $v_i$ relabels each vertex $u \in N(v_i)$ with its own index $i$ and is then removed from the graph.
In the external memory setting we cannot access labels of vertices and relabel them on-the-fly but rather postpone the comparison of the labels to the adjacent vertices instead.
To do so, $v_i$ forwards its own label $\ell(v_i)$ to $u \in H(v_i)$ by sending two messages $\langle u, v_i, \ell(v_i) \rangle$ and $\langle u, v_i, i \rangle$ to $u$, signaling that $u$ should compare its own label to $v_i$'s label $\ell(v_i)$ and then update it to $i$.
Since the label of any adjacent vertex is changed after processing a vertex, when arriving at vertex $v_j$ an odd number of messages will be targeted to $v_j$, where the last one corresponds to its actual label at step $j$.
Then, after collecting all received labels, we compare disjoint consecutive pairs of labels and check whether they match.
In the membership case, we do not find any mismatch and return $\gamma$ as the \ycert.
Otherwise, we have to return a $P_4$ or $C_4$.

In the description of \cite{HK07} the authors stop at the first anomaly where $v_i$ detects a mismatch in its own label and one of its neighbors.
We simulate the same behavior by writing out every anomaly we find, e.g.~that $v_j$ does not have the expected label of $v_i$ via an entry $\langle v_i, v_j, k\rangle$ where $k$ denotes the found label of $v_j$.
After sorting the entries, we find the earliest anomaly $\langle v_i, v_j, k\rangle$ with the largest label $k$ of $v_i$'s neighbors.
Since $v_j$ received the label $k$ from $v_k$, but $v_i$ did not, it is clear that $v_k$ is not universal in its connected component in $G[\{v_k, v_{k+1}, \ldots, v_n\}]$ and we thus return a $P_4$ or $C_4$.
Note that $(v_k, v_j, v_i)$ already constitutes a $P_3$ where $\deg(v_k) \ge \deg(v_j)$, because $v_j$ received the label $k$.
Since $v_j$ is adjacent to both $v_k$ and $v_i$ and $\deg(v_k) \ge \deg(v_j)$, there must exist a vertex $x \in N(v_k)$ where $\{v_j, x\} \notin E$.
Thus, $G[\{ v_k, v_j, v_i, x \}]$ is a $P_4$ if $\{v_i, x\} \notin E$ and a $C_4$ otherwise.
Finding $x$ and determining whether the forbidden substructure is a $P_4$ or a $C_4$ requires scanning $\OO{1}$ adjacency lists in $\OO{\scan(n)}$ I/Os.

\begin{algorithm}[tp]
		\DontPrintSemicolon
		\caption{Recognizing Threshold Graphs for Split Graphs in EM}
		\label{algo:p4threshold}
		\KwData{edges $E$ of split graph $G$, max.~clique size $k$, peo $\alpha = (v_1, \ldots, v_n)$}
		\KwOut{bool whether $G$ is threshold}
		Sort $E$ and relabel according to $\alpha$ \\
		Vector \textsf{S} \\
		\For{$i = 1, \ldots, n$}{
			\If{$L(v_i) = \emptyset$}{
				Let $v_j$ be the smallest successor of $v_i$ in $H(v_i)$ \\
				\textsf{S.push}($v_j$) \tcp*{$v_i$ would be isolated after deleting $\{v_j, \ldots, v_n \}$}
			}
		}
		Sort \textsf{S} in non-ascending order\\
		$n_{\text{del}} \leftarrow 0$ \tcp*{number of deleted universal/isolated vertices}
		\For{$i = n, \ldots, 1$}{
			\If(\hfill \CommentSty{// $v_i$ not isolated in $G[\{ v_1, \ldots, v_n \}]$}){$L(v_i) \ne \emptyset$}{
				\If(\hfill \CommentSty{// $v_i$ not universal}){$|L(v_i)| < (n - 1) - n_{\text{del}}$}{
					\Return \textsc{false}
				}
				$n_{\text{del}} \leftarrow n_{\text{del}} + 1 + \text{occurrences of } v_i$ \tcp*{$v_i$ removed, scan \textsf{S}}
			}
		}
		\Return \textsc{true}
	\end{algorithm}


\begin{proposition}
	\label{prop:uco}
	Verifying that a non-increasing degree ordering $\gamma = (v_1, \ldots, v_n)$ of a graph $G$ with $n$ vertices and $m$ edges is a universal-in-a-component-ordering requires $\OO{\sort(m)}$ I/Os.
\end{proposition}
\begin{proof}
	Every vertex $v_i$ receives exactly two messages per neighbor in $L(v_i)$ and verifies that all consecutive pairs of labels match.
	Then, either the label $i$ is sent to each higher ranked neighbor of $H(v_i)$ via TFP or it is verified that $\gamma$ is not a uco.
	Since at most $\OO{m}$ messages are inserted, the resulting overall complexity is $\OO{\sort(m)}$ I/Os.
	Correctness follows from \cite{HK07} since \autoref{algo:uco} performs the same operations but only delays the label comparisons. \hfill \qed
\end{proof}

We again summarize our results in \autoref{lem:triviallyperfect}.

\begin{lemma}
	\label{lem:triviallyperfect}
	A graph with $n$ vertices and $m$ edges stored in external memory is certified whether it is a trivially perfect graph or not in $\OO{\sort(n + m)}$ I/Os.
	In the membership case the algorithm returns in $\OO{\scan\left(\gamma \right)}$ I/Os the universal-in-a-component ordering $\gamma$ as the \ycert, and otherwise it returns a $\OO{1}$-size \ncert.
\end{lemma}

\subsection{Certifying Bipartite Chain Graphs in External Memory}

Bipartite chain graphs are bipartite graphs where one part of the bipartition has an nno \cite{Y81} similar to threshold graphs.
Interestingly, for chain graphs one side of the bipartition exhibits this property if and only if both partitions do \cite{Y81}.
Its forbidden induced substructures are $2K_2, C_3$ and $C_5$.
By definition, bipartite chain graphs are bipartite graphs which therefore requires I/O-efficient bipartiteness testing.

We follow the linear time internal memory approach of \cite{HK07} with slight adjustments to accommodate the external memory setting.
First, we check whether the input is indeed a bipartite graph.
Instead of using breadth-first search which is very costly in external memory, even for constrained settings \cite{AM09}, we can use a more efficient approach with spanning trees which is presented in \autoref{lem:cycle}.
In case the input is not connected, we simply return two edges of two different components as the $2K_2$.
If the graph is connected, we proceed to verify that the graph is bipartite and return a \ncert\ in the form of a $C_3, C_5$ or $2K_2$ in case it is not.
In order to find a $C_3, C_5$ or $2K_2$ some modifications to \autoref{lem:cycle} are necessary.
Essentially, the algorithm instead returns a minimum odd cycle that is built from $T$ and a single non-tree edge.
Due to minimality we can then find a $2K_2$.
The result is summarized in \autoref{cor:cycle2}.

Then, it remains to show that each side of the bipartition has an nno.
Let~$U$ be the larger side of the partition.
By \cite{MP95} it suffices to show that the input is a chain graph iff the graph obtained by adding all possible edges with both endpoints in $U$ is a threshold graph.
Instead of materializing the mentioned threshold graph, we implicitly represent the adjacencies of vertices in $U$ to retain the same I/O-complexity and apply \autoref{lem:threshold} using $\OO{\sort(n + m)}$ I/Os.
If the input is bipartite but not chain, we repeatedly delete vertices that are connected to all other vertices of the other side and the resulting isolated vertices, similar to \autoref{subsec:triviallyperfect} and \cite{HK07}.
After this, the vertex $v$ with highest degree has a non-neighbor $y$ in the other partition.
By similar arguments $y$ is adjacent to another vertex $z$ that is adjacent to a vertex $x$ where $\{v, x\} \notin E$ \cite{HK07}.
As such $G[\{v,y,z,x\}]$ is a $2K_2$ and can be found in $\OO{\scan(n)}$ I/Os and returned as the \ncert.

\begin{lemma}
	\label{lem:cycle}
	A graph with $n$ vertices and $m$ edges stored in external memory is certified whether it is a bipartite graph or not in $\OO{\sort(n + m)}$ I/Os, given a spanning forest of the input graph.
	In the membership case the algorithm returns in $\OO{scan\left(n\right)}$ I/Os a bipartition $(U, V \setminus U)$ as the \ycert, and otherwise it returns an odd cycle as the \ncert.
\end{lemma}
\begin{proof}
	In case there are multiple connected components, we operate on each individually and thus assume that the input is connected.
	Let $T$ be the edges of the spanning tree and $E \setminus T$ the non-tree edges.
	Any edge $e \in E \setminus T$ may produce an odd cycle by its addition to $T$.
	In fact, the input is bipartite iff $T \cup \{ e \}$ is bipartite for all $e \in E \setminus T$\footnote{Since $T$ is bipartite, one can think of $T$ as a representation of a $2$-coloring on $T$.}.
	We check whether an edge $e = \{u, v\}$ closes an odd cycle in $T$ by computing the distance $d_T(u, v)$ of its endpoints in $T$.
	Since this is required for every non-tree edge $E \setminus T$, we resort to batch-processing.
	Note that $T$ is a tree and hence after choosing a designated root $r \in V$ it holds that $d_T(u, v) = d_T(u, \LCA(u, v)) + d_T(v, \LCA(u, v))$ where $\LCA(u, v)$ is the lowest common ancestor of $u$ and $v$ in $T$.
	Therefore for every edge $E \setminus T$ we compute its lowest common ancestor in $T$ using $\OO{(m/n)\cdot\sort(n)} = \OO{\sort(m)}$ I/Os \cite{CGGTVV95}.
	
	Additionally, for each vertex $v \in V$ we compute its depth in $T$ in $\OO{\sort(m)}$ I/Os using Euler Tours \cite{CGGTVV95} and inform each incident edge of this value by a few scanning and sorting steps.
	Similarly, each edge $e = \{u, v\}$ is provided of the depth of $\LCA(u, v)$.
	Then, after a single scan over $E \setminus T$ we compute $d_T(u, v)$ and check if it is even.
	If any value is even, we return the odd cycle as a \ncert\ or a bipartition in $T$ as the \ycert.
	Both can be computed using Euler Tours in $\OO{\sort(m)}$ I/Os. \hfill \qed
\end{proof}

\begin{corollary}
	\label{cor:cycle2}
	If a connected graph $G$ contains a $C_3, C_5$ or $2K_2$ then any of these subgraphs can be found in $\OO{\sort(n+m)}$ I/Os given a spanning tree of $G$.
\end{corollary}
\begin{proof}
		We extend the algorithm presented in \autoref{lem:cycle} since it does not return an induced cycle.
		While iterating over the edges to find an odd cycle we save the smallest seen odd cycle by keeping a copy of the edge $e \in E \setminus T$ and the length of the minimum odd cycle.
		In case we find a $C_3$ or a $C_5$ we are done and return the \ncert\ immediately otherwise for a $C_k$ with $k = 2\ell + 1 > 5$ we return a $2K_2$ by finding a matching edge to the non-tree edge $e \in E \setminus T$ in the cycle.
		
		Let $C = (u_1, \ldots, u_k, u_1)$ be the returned cycle where $\{u_k, u_1\}$ is the non-tree edge.
		In this case we return for the $2K_2$ the graph ($\{ u_{\ell}, u_{\ell + 1}, u_1, u_k \}$, $\{ \{u_1, u_k\}, \{u_\ell, u_{\ell + 1}\} \}$).
		If $\ell$ is odd, the non-edges of the $2K_2$ cannot exist since otherwise any of the following smaller odd cycles $(u_1, u_2, \ldots, u_{\ell + 1}, u_k, u_1)$, $(u_1, u_2$, $\ldots, u_\ell, u_1)$, $(u_\ell, u_{\ell + 1}, \ldots, u_k, u_\ell)$ and $(u_1, u_{\ell + 1}, u_{\ell + 2}, \ldots, u_k, u_1)$ would be present, contradicting the minimality of $C$.
		For the other case where $\ell$ is even, a similar argument can be found.
		The I/O-complexity therefore remains the same. \hfill \qed
		%
	\end{proof}

We summarize our findings for bipartite chain graphs in \autoref{lem:chain}.

\clearpage

\begin{lemma}
	\label{lem:chain}
	A graph with $n$ vertices and $m$ edges stored in external memory is certified whether it is a bipartite chain graph or not in $\OO{\sort(n + m)}$ I/Os with high probability.
	In the membership case the algorithm returns in $\OO{\scan\left(n\right)}$ I/Os the bipartition $(U, V \setminus U)$ and nested neighborhood orderings of both partitions as the \ycert, and otherwise it returns a $\OO{1}$-size \ncert.
\end{lemma}
\begin{proof}
	Computing a spanning tree $T$ requires $\OO{\sort(n + m)}$ I/Os with high probability by an external memory variant of the Karger, Klein and Tarjan minimum spanning tree algorithm  \cite{CGGTVV95}.
	By \autoref{cor:cycle2} we find a $C_3, C_5$ or $2K_2$ if the input is not bipartite or not connected.
	We proceed by checking the nno's of both partitions in $\OO{\sort(n + m)}$ I/Os using \autoref{lem:threshold}. \hfill \qed
\end{proof}

\section{Experimental Evaluation}
\begin{figure}[t]
	\includegraphics[width=.5\textwidth]{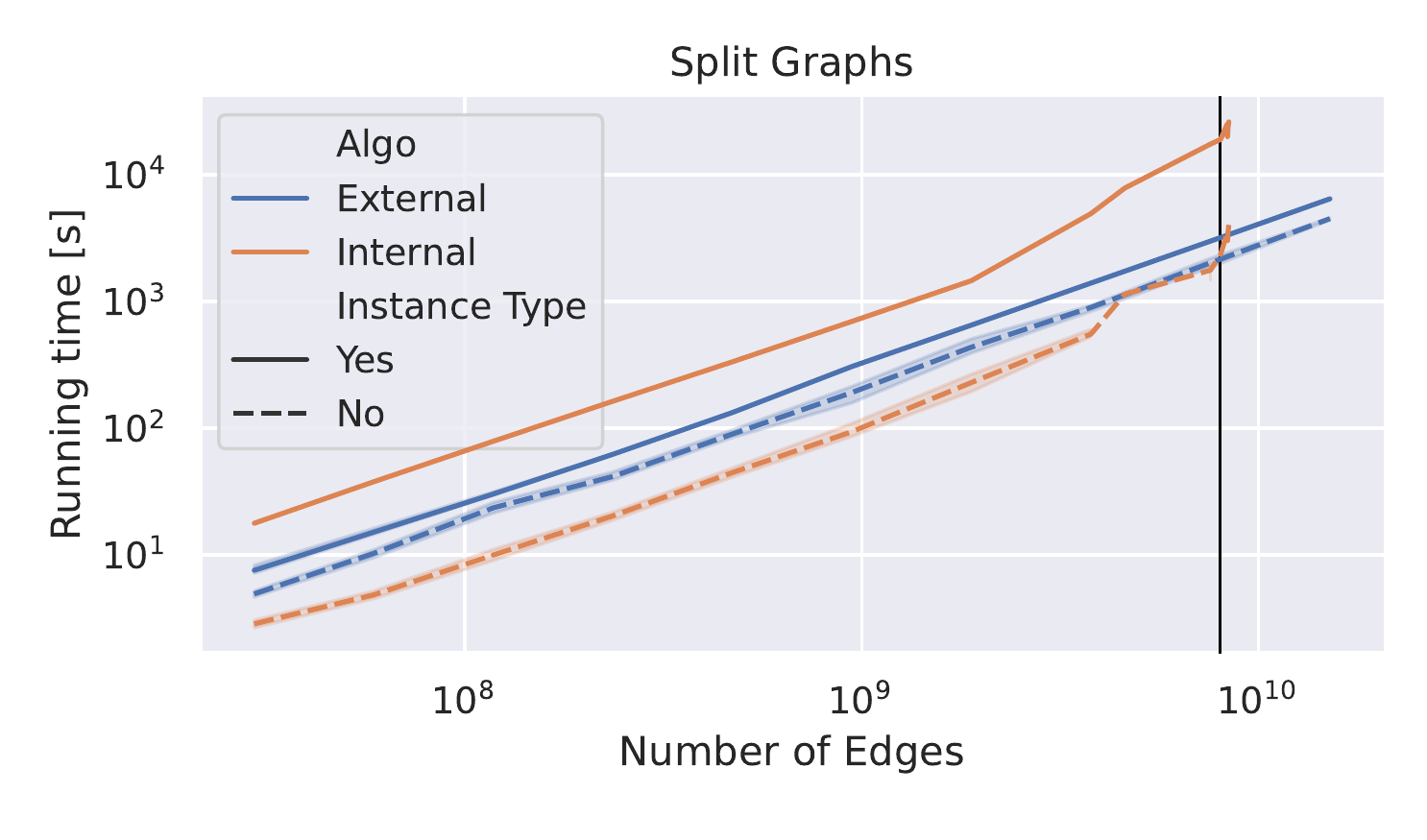}
	\includegraphics[width=.5\textwidth]{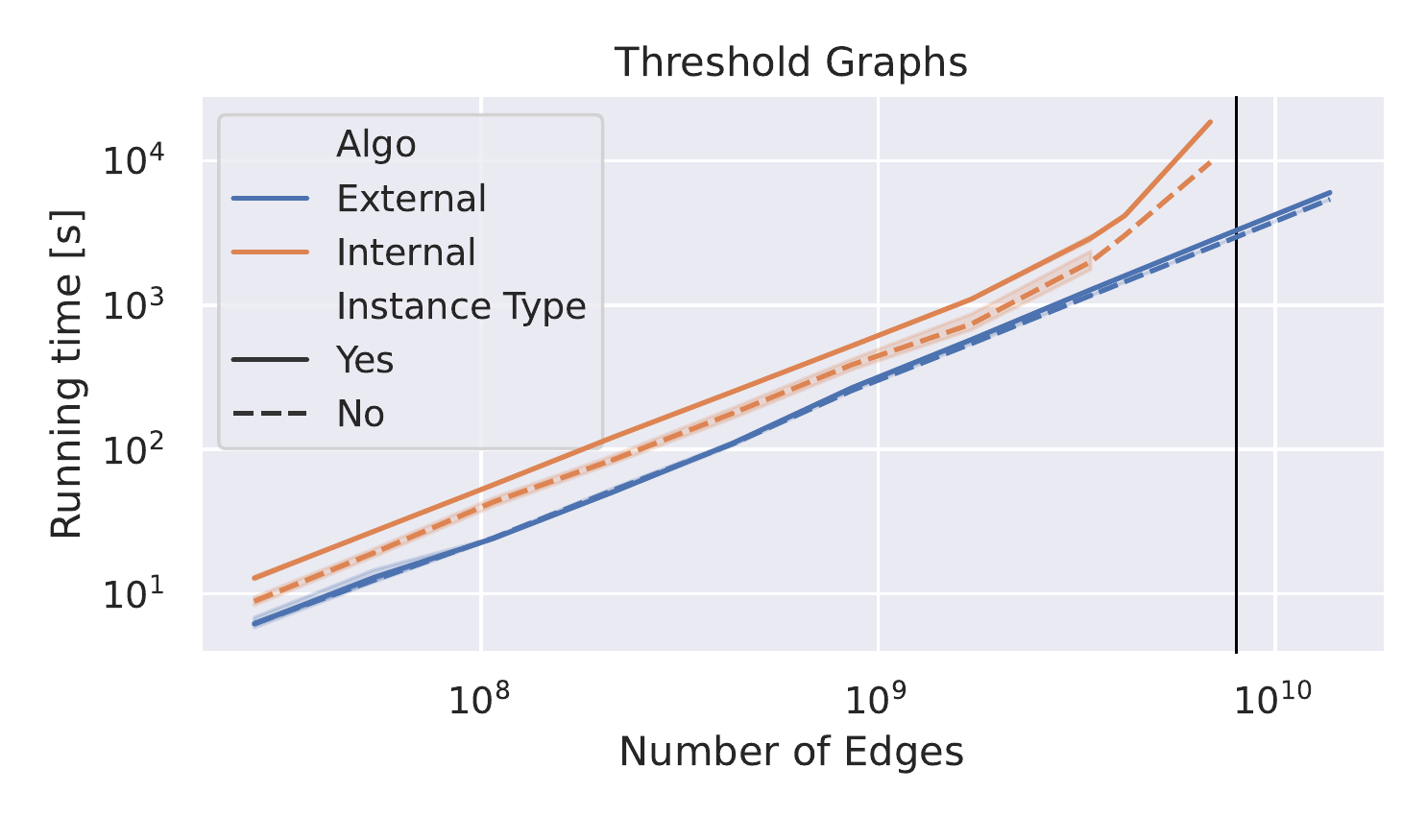}
	\caption{
	Running times of the external memory algorithms for certifying split (left) and threshold graphs (right) for different random graph instances.
	The black vertical lines depicts the number of elements that can concurrently be held in internal memory.
	}
	\label{fig:exp}
\end{figure}

We have implemented our external memory certifying algorithms for split and threshold graphs in \textsf{C++} using the STXXL library~\cite{DKS08}.
To provide a comparison of our algorithms, we also implemented the internal memory state-of-the-art algorithms by Heggernes and Kratsch~\cite{HK07}.
STXXL offers external memory versions of fundamental algorithmic building blocks like scanning, sorting and several data structures.
Our benchmarks are built with \textsf{GNU g++-10.3} and executed on a machine equipped with an AMD EPYC 7302P processor and 64 GB RAM running Ubuntu 20.04 using six 500 GB solid-state disks.

In order to validate the predicted scaling behaviour we generate our instances parameterized by $n$.
For \textsc{yes}-instances of split graphs we generate a split partition $(K, I)$ with $|K| = n/10$ and add each possible edge $\{u, v\}$ with probability $1/4$ for $u \in I$ and $v \in K$.
Analogously, \textsc{yes}-instances of threshold graphs are generated by repeatedly adding either isolated or universal vertices with probability $9/10$ and $1/10$, respectively.
We additionally attempt to generate \textsc{no}-instances by adding $\OO{1}$ many random edges to the \textsc{yes}-instances.
In a last step we randomize the vertex indices to extend the effect of random accesses on the running time of the algorithms.

In \autoref{fig:exp} we present the running times of all algorithms on multiple \textsc{yes}- and \textsc{no}-instances.
It is clear that the performance of both external memory algorithms is not impacted by the main memory barrier while the running time of their internal memory counterparts already increases when at least half the main memory is used.
This effect is amplified immensely after exceeding the size of main memory by only a small fraction for split graphs, \autoref{fig:exp} (left) and we expect the same for threshold graphs.

Certifying the produced \textsc{no}-instances of split graphs seems to require less time than their corresponding unmodified \textsc{yes}-instances as the algorithm typically stops prematurely.
Furthermore, due to the low data locality of the internal memory variant it is apparent that the external memory algorithm is superior for the \textsc{yes}-instances.
The performance on both \textsc{yes}- and \textsc{no}-instances is very similar in external memory.
This is in part due to the fact that the algorithm first performs a relabeling which increases the ratio of common computation significantly.

For threshold graphs, the external memory variant outperforms the internal memory variant due to improved data locality.
Analogously to split graphs, the difference in performance between \textsc{yes}- and \textsc{no}-instances is more profound for the internal memory variants.

\section{Conclusions}
We have presented the first I/O-efficient certifying recognition algorithms for split, threshold, trivially perfect, bipartite and bipartite chain graphs.
Our algorithms require $\OO{\sort(n + m)}$ I/Os matching common lower bounds for many algorithms in external memory.
It would be interesting to further extend the scope of certifying algorithms to more graph classes for the external memory regime.


\bibliographystyle{splncs04}
\bibliography{certi}

\clearpage

\appendix

\begin{section}{Appendix}
\label{sec:appendix}
	
	\subsection{Further Discussion on the Returned Certificates}
	We note that reverting any relabeling again requires only $\OO{\sort(n + m)}$ I/Os by a constant number of scanning and sorting steps.
	Authenticating the \ycert{}s of all our algorithms requires $\OO{\sort(n + m)}$ I/Os anyway which is why we assume that the graph is given in its relabeled form.
	
	\begin{proposition}
	    \label{prop:auth-split}
	    Authenticating $(K, I)$ for a given split graph with $n$ vertices and $m$ edges requires $\OO{\sort(n + m)}$ I/Os.
	\end{proposition}
	\begin{proof}
	    Since we can assume relabelled vertex indices, let $I = \{v_1, \ldots,  v_k\}$ and $K = V \setminus I$.
	    After sorting the edges in $\OO{\sort(m)}$ I/Os, we check that no edge between $I$ and $K$ exists by comparing the indices.
	    Verifying that $K$ is a clique only requires looking at the $\binom{|K|}{2}$ last edges where both are done in $\OO{\scan(m)}$ I/Os.
	\end{proof}
	
	\begin{proposition}
	    \label{prop:auth-threshold}
	    Authenticating $\beta = (u_1, \ldots, u_k)$ for a given threshold graph with $n$ vertices and $m$ edges requires $\OO{\sort(n + m)}$ I/Os.
	\end{proposition}
	\begin{proof}
	    We can again assume that the vertices in the ordering of $\beta$ are given by $I = \{v_1, \ldots, v_k\}$.
	    Verifying $(K, I)$ is done as described in \autoref{prop:auth-split} using $\OO{\sort(m)}$ I/Os.
	    It remains to verify that $\beta$ is a nno.
	    For increasing $i$ we verify $N(v_i) \subseteq N(v_{i+1})$ by a concurrent scan over both neighborhoods requiring in total $\OO{\scan(m)}$ I/Os for all $i$.
	\end{proof}
	
	\begin{proposition}
	    Authenticating $\gamma = (v_1, \ldots, v_n)$ for a given trivially perfect graph with $n$ vertices and $m$ edges requires $\OO{\sort(n + m)}$ I/Os.
	\end{proposition}
	\begin{proof}
	    We rerun the \autoref{algo:p4threshold} using $\OO{\sort(n + m)}$ I/Os as the certificate is the ordering itself.
	\end{proof}
	
	\begin{proposition}
	    Authenticating $(U, V\setminus U)$ with two nested neighborhood orderings for a given bipartite chain graph with $n$ vertices and $m$ edges requires $\OO{\sort(n + m)}$ I/Os.
	\end{proposition}
	\begin{proof}
	    Similar to \autoref{prop:auth-split} we check that $U$ and $V \setminus U$ are both independent sets using $\OO{\sort(n + m)}$ I/Os.
	    Thereafter similar to \autoref{prop:auth-threshold} we verify that both orderings are indeed nested neighborhood orderings using again $\OO{\sort(n + m)}$ I/Os.
	\end{proof}
\end{section}

\end{document}